\documentclass[a4paper, USenglish, cleveref, autoref, thm-restate]{lipics-v2021}
\usepackage{complexity}
\usepackage{tikz}

\pdfoutput=1 
\hideLIPIcs  
\nolinenumbers

\DeclareMathOperator{\AND}{AND}
\DeclareMathOperator{\OR}{OR}
\DeclareMathOperator{\PAR}{PAR}
\DeclareMathOperator{\size}{size}

\tikzstyle{every picture} = [>=latex]

\title{CNF Encodings of Parity}

\author{Gregory Emdin}{St.~Petersburg State University}{egd3700@mail.ru}{}{}

\author{Alexander~S. Kulikov}{Steklov Mathematical Institute at St.~Petersburg, Russian Academy of Sciences \and St.~Petersburg State University \and \url{https://logic.pdmi.ras.ru/~kulikov/}}{kulikov@logic.pdmi.ras.ru}{https://orcid.org/0000-0002-5656-0336}{}

\author{Ivan Mihajlin}{Steklov Mathematical Institute at St.~Petersburg, Russian Academy of Sciences}{ivmihajlin@gmail.com}{}{}

\author{Nikita Slezkin}{Steklov Mathematical Institute at St.~Petersburg, Russian Academy of Sciences}{ne.slezkin@gmail.com}{https://orcid.org/0000-0003-1904-9261}{}

\authorrunning{G.~Emdin, A.~Kulikov, I.~Mikhajlin, N.~Slezkin} 

\Copyright{Gregory Emdin, Alexander~S. Kulikov, Ivan Mikhajlin, Nikita Slezkin}

\ccsdesc[500]{Theory of computation~Complexity theory and logic}
\ccsdesc[500]{Theory of computation~Circuit complexity}
\keywords{encoding, parity, lower bounds, circuits, CNF} 

\EventEditors{John Q. Open and Joan R. Access}
\EventNoEds{2}
\EventLongTitle{42nd Conference on Very Important Topics (CVIT 2016)}
\EventShortTitle{CVIT 2016}
\EventAcronym{CVIT}
\EventYear{2016}
\EventDate{December 24--27, 2016}
\EventLocation{Little Whinging, United Kingdom}
\EventLogo{}
\SeriesVolume{42}
\ArticleNo{23}

\begin{document}
\sloppy
\maketitle

\begin{abstract}
    The minimum number of clauses in a~CNF representation of the parity function
    $x_1 \oplus x_2 \oplus \dotsb \oplus x_n$ is $2^{n-1}$.
    One can obtain a~more compact CNF encoding by~using non-deterministic variables
    (also known as~guess or~auxiliary variables). In~this paper,
    we prove the following lower bounds, that almost match known upper bounds,
    on~the number~$m$ of~clauses and the maximum width~$k$ of clauses:
    1)~if there are at~most $s$~auxiliary variables, then $m \ge \Omega\left(2^{n/(s+1)}/n\right)$ and $k \ge n/(s+1)$;
    2)~the minimum number of clauses is at~least~$3n$.
    We~derive the first two bounds
    from the Satisfiability Coding Lemma due to~Paturi,
    Pudl{\'{a}}k, and Zane
    using a~tight connection
    between CNF encodings and depth-$3$ circuits.
    In~particular, we~show that lower bounds on~the size of~a~CNF encoding
    of~a~Boolean function imply depth-$3$ circuit lower bounds for this function.
\end{abstract}


\section{Overview}

\subsection{Motivation}
A~popular approach for solving a~difficult combinatorial problem in~practice
is to~encode~it in~conjunctive normal form (CNF) and to~invoke a~SAT-solver.
There are two main reasons why this approach works well for many hard problems:
the state-of-the-art SAT-solvers are extremely efficient
and many combinatorial problems are expressed naturally in~CNF.
At~the same time, a~CNF encoding is~not unique and one usually
determines a~good encoding empirically. Moreover, there is
no~such thing as~the best encoding of a~given problem
as~it~also depends on a~SAT-solver at~hand.
Prestwich~\cite{DBLP:series/faia/Prestwich09}
gives an~overview of~various ways to~translate
problems into CNF and discusses their desirable properties,
both from theoretical and practical points of~view.

Already for such simple functions as~the parity function $x_1 \oplus x_2 \oplus \dotsb \oplus x_n$, it is not immediate how to~encode them in CNF
(to~make~it efficiently handled by SAT-solvers).
Parity function is used frequently in cryptography (hash functions, stream ciphers, etc.). It~is known that the minimum number of~clauses in a~CNF computing parity
is~$2^{n-1}$. This becomes impractical quickly as~$n$~grows. A~standard way to~reduce
the size of~an~encoding is by~using non-deterministic variables (also known as~guess
or~auxiliary variables). Namely,
one introduces
$s$~non-deterministic variables $y_1, \dotsc, y_s$
and
partitions
the set of input variables into $s+1$~blocks of size at~most $\lceil n/(s+1) \rceil$: $\{x_1, x_2, \dotsc, x_n\}=X_1 \sqcup X_2 \sqcup \dotsb \sqcup X_{s+1}$. Then, one writes down the following $s+1$ parity functions in~CNF:
\begin{multline}\label{eq:blocks}
    \left(y_1=\bigoplus_{x \in X_1}x\right),
    \left(y_2=y_1 \oplus \bigoplus_{x \in X_2}x\right), \dotsc,\\
    \left(y_s=y_{s-1} \oplus \bigoplus_{x \in X_s}x\right),
    \left(1=y_s \oplus \bigoplus_{x \in X_{s+1}}x\right).
\end{multline}
The value for the parameter~$s$ is~usually determined experimentally.
For example,
Prestwich~\cite{DBLP:journals/dam/Prestwich03}
reports that taking $s=10$ gives the best results when solving
the minimal disagreement parity learning problem using
local search based SAT-solvers.

The simple construction above implies several upper bounds on~the number~$m$ of~clauses,
the number~$s$ of~non-deterministic variables, and the width~$k$ of clauses:
\begin{description}
    \item[Limited non-determinism:] using $s$~non-deterministic variables, one can encode parity either as
    a~CNF with at~most \[m \le (s+1)2^{\lceil n/(s+1) \rceil+2-1} \le 4(s+1)2^{n/(s+1)}\] clauses or as a~$k$-CNF, where \[k=2+{\lceil n/(s+1) \rceil} \le 3+n/(s+1) \, .\]
    \item[Unlimited non-determinism:] one can encode parity as a~CNF with at~most $4n$~clauses (to~do this, use $s=n-1$ non-deterministic variables; then, each of~$n$~functions in~\eqref{eq:blocks} can be~written in~CNF using at~most four clauses).
\end{description}

\subsection{Results}

In~this paper, we~show that
the upper bounds mentioned above are
essentially optimal.
\begin{theorem}\label{thm:main}
    Let $F$~be a~CNF-encoding of $\PAR_n$ with $m$~clauses, $s$~non-deterministic variables, and maximum clause width~$k$.
    \begin{enumerate}
        \item The parameters $s$~and~$m$ cannot be~too small simultaneously:
        \begin{equation}\label{eq:sm}
            m \ge \Omega\left(\frac{s+1}{n} \cdot 2^{n/(s+1)}\right) \, .
        \end{equation}
        \item The parameters $s$~and~$k$ cannot be~too small simultaneously:
        \begin{equation}\label{eq:sw}
            k \ge n/(s+1) \, .
        \end{equation}
        \item The parameter~$m$ cannot be~too small:
        \begin{equation}\label{eq:m}
            m \ge 3n-9 \, .
        \end{equation}
    \end{enumerate}
\end{theorem}

\subsection{Techniques}
We~derive a~lower bound $m \ge \Omega((s+1)2^{n/(s+1)}/n)$ from the
Satisfiability Coding Lemma due~to Paturi,
Pudl{\'{a}}k, and Zane~\cite{DBLP:journals/cjtcs/PaturiPZ99}. This lemma allows
to~prove a~$2^{\sqrt n}$ lower bound on~the size of~depth-$3$ circuits
computing the parity function. Interestingly, the lower bound $m \ge \Omega((s+1)2^{n/(s+1)}/n)$
implies a~lower bound $2^{\Omega(\sqrt n)}$ almost immediately, though it~is
not clear whether a~converse implication can be~easily proved.

To~prove a lower bound $m \ge 3n-9$, we~analyze carefully the structure of a~CNF encoding.

\subsection{Related work}
Many results for various computational models with limited non-determinism
are surveyed by~Goldsmith, Levy, and Mundhenk~\cite{DBLP:journals/sigact/GoldsmithLM96}.
An~overview of known approaches for CNF encodings is~given
by~Prestwich~\cite{DBLP:series/faia/Prestwich09}. Two recent results that are close
to~the results of~this paper are the following.
Morizumi~\cite{DBLP:conf/cocoon/Morizumi15} proved that non-deterministic inputs
do~not help in~the model of~Boolean circuits over the $U_2$~basis (the set of~all binary functions except for the binary parity and its complement) for computing the parity function: with and without non-deterministic inputs, the minimum size
of~a~circuit computing parity is $3(n-1)$.
Kucera, Savick{\'{y}}, Vorel~\cite{DBLP:journals/tcs/KuceraSV19} prove almost tight bounds on~the size of~CNF encodings of the at-most-one Boolean function ($[x_1+\dotsb+x_n \le 1]$).
Sinz~\cite{DBLP:conf/cp/Sinz05} proves a~linear lower bound on~the size of~CNF encodings of the at-most-$k$ Boolean function.

\section{General setting}

\subsection{Computing Boolean functions by~CNFs}
For a~Boolean function $f(x_1, \dotsc, x_n) \colon \{0,1\}^n \to \{0,1\}$, we say that a~CNF~$F(x_1, \dotsc, x_n)$ \emph{computes~$f$} if $f \equiv F$, that is, for all $x_1, \dotsc, x_n \in \{0,1\}$, $f(x_1, \dotsc, x_n)=F(x_1, \dotsc, x_n)$.
We~treat a~CNF as a~set of clauses and by~the \emph{size}
of a~CNF we~mean its number of clauses. It~is well known that for every function~$f$, there exists a~CNF computing~it. One way to~construct such a~CNF is the following: for every input $x \in \{0,1\}^n$ such that $f(x)=0$, populate a~CNF with a~clause of length~$n$ that is falsified by~$x$.

This method does not guarantee that the produced CNF has the minimal number
of~clauses: this would~be too good to be~true as~the problem of~finding a~CNF
of~minimum size for a~given Boolean function
(specified by~its truth table) is NP-complete as~proved by~Masek~\cite{MasekNpComp} (see also \cite{DBLP:journals/siamcomp/AllenderHMPS08} and references herein).
For example, for a~function $f(x_1,x_2)=x_1$ the method
produces a~CNF $(\overline{x_1} \lor x_2) \land (\overline{x_1} \lor \overline{x_2})$ whereas the function $x_1$ is~already in~CNF format.

\subsection{Parity}
It is well known that for many functions, the minimum size of a~CNF is exponential. The canonical example is the parity function $\PAR_n(x_1, \dotsc, x_n)=x_1 \oplus \dotsb \oplus x_n$. The property of $\PAR_n$ that prevents it from being computable
by~short CNF's is~its high \emph{sensitivity}: by~flipping \emph{any}
bit in \emph{any} input~$x \in \{0,1\}^n$, one flips the value
of~$\PAR_n(x)$.

\begin{lemma}\label{lemma:detparity}
    The minimum size of a~CNF computing~$\PAR_n$ has size $2^{n-1}$.
\end{lemma}
\begin{proof}
    An~upper bound follows from the method above by~noting that $|\PAR_n^{-1}(0)|=2^{n-1}$.

    A~lower bound is based on~the fact that any clause of a~CNF~$F$ computing $\PAR_n$ must contain all variables $x_1, \dotsc, x_n$. Indeed,
    if a~clause $C \in F$ did not depend on~$x_i$, one could find
    an~input $x \in \{0,1\}^n$ that falsifies~$C$ (hence, $F(x)=\PAR_n(x)=0$) and remains to~be falsifying even after flipping~$x_i$. As~any clause of~$F$ has exactly $n$~variables, it~rejects exactly one $x \in \{0,1\}^n$. Hence, $F$~must contain at~least $|\PAR_n^{-1}(0)|=2^{n-1}$ clauses.
\end{proof}

\subsection{Encoding Boolean functions by~CNFs}
We say that a~CNF~$F$ \emph{encodes} a~Boolean function $f(x_1, \dotsc, x_n)$ if the following two conditions hold.
\begin{enumerate}
    \item In~addition to the input bits $x_1, \dotsc, x_n$, $F$~also depends on $s$~bits $y_1, \dotsc, y_s$ called \emph{guess inputs} or \emph{non-deterministic inputs}.
    \item For every $x \in \{0,1\}^n$, $f(x)=1$ iff there exists $y \in \{0,1\}^s$ such that $F(x,y)=1$. In other words, for every $x \in \{0,1\}^n$,
    \begin{equation}\label{eq:enc}
        f(x) = \bigvee_{y \in \{0,1\}^s}F(x,y) \, .
    \end{equation}
\end{enumerate}
Such representations of Boolean functions are widely used in~practice
when one translates a~problem to~SAT.
For example, the following CNF encodes
$\PAR_4$:
\begin{multline}\label{eq:toyenc}
    (x_1 \lor x_2 \lor \overline{y_1}) \land (x_1 \lor  \overline{x_2} \lor y_1) \land (\overline{x_1} \lor x_2 \lor y_1) \land (\overline{x_1} \lor \overline{x_2} \lor \overline{y_1})
    \land
    (y_1 \lor x_3 \lor \overline{y_2}) \land\\ (y_1 \lor  \overline{x_3} \lor y_2) \land (\overline{y_1} \lor x_3 \lor y_2) \land (\overline{y_1} \lor \overline{x_3} \lor \overline{y_2})
    \land (\overline{x_4} \lor y_2) \land (x_4 \lor \overline{y_2}) \, .
\end{multline}

\subsection{Boolean Circuits and Tseitin Transformation}
A~natural way to~get a~CNF encoding of~a~Boolean function~$f$ is~to take a~circuit
computing~$f$ and apply Tseitin transformation~\cite{zbMATH03325539}.
We~describe this transformation using a~toy example.
The following circuit computes $\PAR_{12}$ with three gates.
It has $12$~inputs, $3$~gates (one of which is an~output gate), and has depth~$3$.

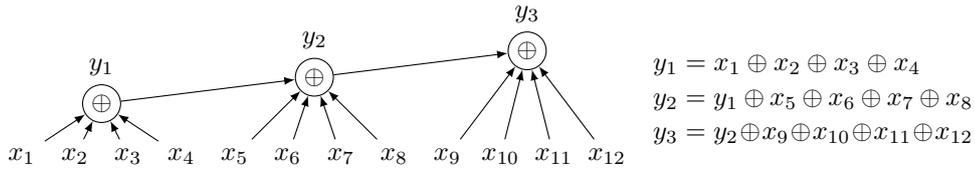
\tikzstyle{gate} = [circle, draw, inner sep=0mm, minimum size=5mm]

\begin{center}
\begin{tikzpicture}[yscale=1]
    \begin{scope}[xscale=0.7, yscale=0.7]
    \foreach \n in {1,...,12}
        \node (\n) at (\n,0) {$x_{\n}$};
    \foreach \x/\y/\n/\l/\edges in {2.5/1/y1/y_1/{1,2,3,4}, 6.5/1.5/y2/y_2/{5,6,7,8,y1}, 10.5/2/y_3/y_3/{9,10,11,12,y2}} {
        \node[gate,label=above:$\l$] (\n) at (\x,\y) {$\oplus$};
        \foreach \i in \edges
            \draw[->] (\i) -- (\n);
    }
    \end{scope}
    \node[right, text width=42mm, inner sep=0mm] at (9,.75) {
        $y_1=x_1 \oplus x_2 \oplus x_3 \oplus x_4$\\
        $y_2=y_1 \oplus x_5 \oplus x_6 \oplus x_7 \oplus x_8$\\
        $y_3=y_2 \oplus x_9 \oplus x_{10} \oplus x_{11} \oplus x_{12}$
    };
\end{tikzpicture}
\end{center}
To~the right of~the circuit, we~show
the functions computed by~each gate.
One can translate each line into CNF.
Adding a~clause $(y_3)$ to~the resulting CNF gives a~CNF encoding of the
function computed by~the circuit. In~fact, the CNF~\eqref{eq:blocks}
can be~obtained this way (after propagating the value of~the output gate).

A~CNF can be~viewed as a~depth-$2$ circuit where the output gate is an~AND, all
other gates are~ORs, and the inputs are variables and their negations. For example, the following circuit corresponds to~a~CNF~\eqref{eq:toyenc}. Such depth-2 circuits are also denoted as $\AND \circ \OR$ circuits.

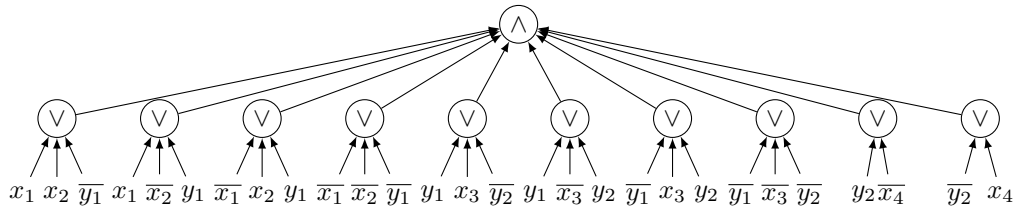
\begin{center}
    \begin{tikzpicture}[xscale=.9, yscale=.5, >=latex]

        \foreach \n\s in {1/x_1, 2/x_2, 3/\overline{y_1},
            4/x_1, 5/\overline{x_2}, 6/y_1,
            7/\overline{x_1}, 8/x_2, 9/y_1,
            10/\overline{x_1}, 11/\overline{x_2}, 12/\overline{y_1},
            13/y_1, 14/x_3, 15/\overline{y_2},
            16/y_1, 17/\overline{x_3}, 18/y_2,
            19/\overline{y_1}, 20/x_3, 21/y_2,
            22/\overline{y_1}, 23/\overline{x_3}, 24/\overline{y_2}}
        \node (v_\n) at (0.5 * \n - 0.5, 0) {$\s$};

        \node (v_25) at (12.3, 0) {$y_2$};
        \node (v_26) at (12.7, 0) {$\overline{x_4}$};
        \node (v_27) at (13.7, 0) {$\overline{y_2}$};
        \node (v_28) at (14.3, 0) {$x_4$};

        \foreach \m in {1,...,10}
            \node[gate] (y\m) at (-1 + 1.5 * \m, 2) {$\lor$};

        \foreach \n in {2,...,25} {
            \def\qq{\the\numexpr\n / 3}
            \def\tt{\the\numexpr\qq}
            \def\dec{\the\numexpr\n - 1}
            \draw[->] (v_\dec) -- (y\tt);
        }

        \draw[->] (v_25) -- (y9);
        \draw[->] (v_26) -- (y9);
        \draw[->] (v_27) -- (y10);
        \draw[->] (v_28) -- (y10);

        \node[gate] (z) at (7.25, 4.5) {$\land$};

        \foreach \n in {1,...,10}
            \draw[->] (y\n) -- (z);

    \end{tikzpicture}
\end{center}

\subsection{Depth-$3$ circuits}
Depth-$3$ circuits is a~natural generalization of CNFs:
a~\emph{$\Sigma_3$-circuit} is~simply an~OR of~CNFs.
In a~circuit, these CNFs are allowed to~share clauses.
A~\emph{$\Sigma_3$-formula} is a~$\Sigma_3$-circuit
whose CNFs do~not share clauses (in~other words, it is a~circuit
where the out-degree of every gate is equal to~one).

On the one hand, this computation model is~still simple enough.
On~the other hand,
proving lower bounds against this model is~much harder: getting
a~$2^{\omega(n)}$ lower bound for an~explicit function (say, from $\NP$ or $\E^{\NP}$)
is a~major challenge. Proving a~lower bound $2^{\omega(n/\log \log n)}$
would resolve another open question, through Valiant's depth reduction~\cite{DBLP:conf/mfcs/Valiant77}: proving a~superlinear lower bound
on~the size of logarithmic depth circuits. We~refer the reader to~Jukna's book~\cite[Chapter~11]{DBLP:books/daglib/0028687} for an~exposition
of~known results for depth-$3$ circuits. For the parity function,
the best known lower bound on~depth-$3$ circuits is~$\Omega(2^{\sqrt{n}})$~\cite{DBLP:journals/cjtcs/PaturiPZ99}.
If one additionally requires that a~circuit is a~formula, i.e., that every gate
has out-degree at~most~1, then the best lower bound is~$\Omega(2^{2\sqrt{n}})$~\cite{DBLP:journals/eccc/Hirahara17}.
Both lower bounds are tight up~to~polynomial factors.

Equation~\eqref{eq:enc} shows a~tight connection between CNF encodings and depth-$3$ circuits of~type $\OR \circ \AND \circ \OR$. Namely, let $F(x_1, \dotsc, x_n, y_1, \dotsc, y_s)=\{C_1, \dotsc, C_m\}$ be
a~CNF encoding of a~Boolean function $f \colon \{0,1\}^n \to \{0,1\}$. Then,
$f(x)=\lor_{y \in \{0,1\}^s}F(x,y)$. By~assigning $y$'s in all $2^s$ ways, one gets
an~$\Sigma_3$-formula that computes~$f$:
\begin{equation}\label{eq:expansion}
    f(x)=\bigvee_{j \in [2^s]}F_j(x) \, ,
\end{equation}
where each $F_j$ is a~CNF. We~call this an~\emph{expansion} of~$F$. For example,
an~expansion of a~CNF~\eqref{eq:toyenc} looks as~follows. It is an~OR of~four CNFs.

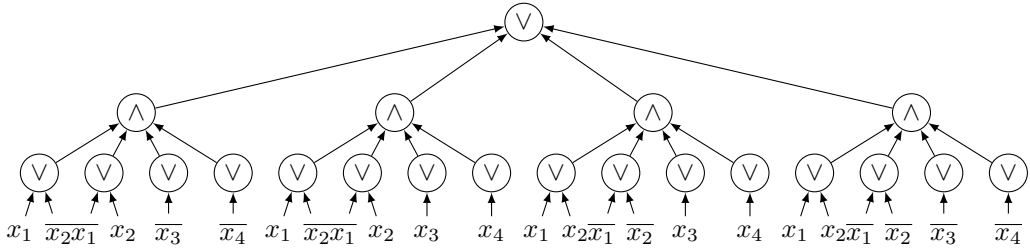
\begin{center}
    \begin{tikzpicture}[xscale=.85, yscale=.4]

        \foreach \n\s\dx\t in {1/x_1/-0.3/0, 2/\overline{x_2}/0.3/0, 3/\overline{x_1}/-0.3/1,
            4/x_2/0.3/1, 5/\overline{x_3}/0/2, 6/\overline{x_4}/0/3,
            7/x_1/-0.3/4, 8/\overline{x_2}/0.3/4, 9/\overline{x_1}/-0.3/5,
            10/x_2/0.3/5, 11/x_3/0/6, 12/x_4/0/7,
            13/x_1/-0.3/8, 14/x_2/0.3/8, 15/\overline{x_1}/-0.3/9,
            16/\overline{x_2}/0.3/9, 17/x_3/0/10, 18/x_4/0/11,
            19/x_1/-0.3/12, 20/x_2/0.3/12, 21/\overline{x_1}/-0.3/13,
            22/\overline{x_2}/0.3/13, 23/\overline{x_3}/0/14, 24/\overline{x_4}/0/15
        }
        \node (v_\n) at (\t+\dx, 0) {$\s$};

        \foreach \m in {1,...,16}
        \node[circle, draw, inner sep=0mm, minimum size=5mm] (y\m) at (-1 + 1 * \m, 2) {$\lor$};

        \foreach \m in {1,...,4}
        \node[circle, draw, inner sep=0mm, minimum size=5mm] (z\m) at (-2.5 + 4 * \m, 4) {$\land$};

        \node[circle, draw, inner sep=0mm, minimum size=5mm] (q) at (7.5, 7) {$\lor$};

        \foreach \n in {1,...,4}
        \draw[->] (z\n) -- (q);

        \foreach \n\m in {1/1, 2/1, 3/1, 4/1, 5/2, 6/2, 7/2, 8/2,
            9/3, 10/3, 11/3, 12/3, 13/4, 14/4, 15/4, 16/4}
        \draw[->]  (y\n) -- (z\m);

        \foreach \n\m in {1/1, 2/1, 3/2, 4/2, 5/3, 6/4,
            7/5, 8/5, 9/6, 10/6, 11/7, 12/8,
            13/9, 14/9, 15/10, 16/10, 17/11, 18/12,
            19/13, 20/13, 21/14, 22/14, 23/15, 24/16}
        \draw[->] (v_\n) -- (y\m);
    \end{tikzpicture}
\end{center}
An~expansion is a~formula: it~is an~OR of~CNFs, every gate has out-degree one.
One can also get a~\emph{circuit-expansion}: in this case, gates are allowed
to~have out-degree more than one; alternatively, CNFs are allowed to~share clauses.
For example, this is a~circuit-expansion of~\eqref{eq:toyenc}.

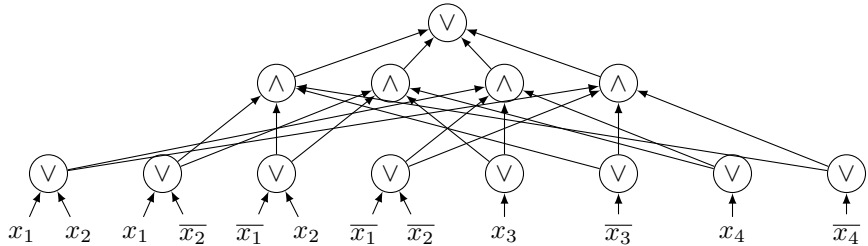
\begin{center}
    \begin{tikzpicture}[yscale=.4]

        \foreach \n\s in {1/x_1, 2/x_2,
            3/x_1, 4/\overline{x_2},
            5/\overline{x_1}, 6/x_2,
            7/\overline{x_1}, 8/\overline{x_2}}
        \node (v_\n) at (0.75 * \n-0.6, 0) {$\s$};
        \foreach \n\s in {9/x_3, 10/\overline{x_3},
            11/x_4, 12/\overline{x_4}}
        \node (v_\n) at (1.5 * \n-7, 0) {$\s$};

        \foreach \m in {1,...,8}
        \node[circle, draw, inner sep=0mm, minimum size=5mm] (y\m) at (-1 + 1.5 * \m, 2) {$\lor$};

        \foreach \m in {1,...,4}
        \node[circle, draw, inner sep=0mm, minimum size=5mm] (z\m) at (2 + 1.5 * \m, 5) {$\land$};

        \node[circle, draw, inner sep=0mm, minimum size=5mm] (q) at (5.75, 7) {$\lor$};

        \foreach \n in {1,...,4}
        \draw[->] (z\n) -- (q);

        \foreach \n\m in {1/3, 1/4, 2/1, 2/2, 3/1, 3/2, 4/3, 4/4,
            5/2, 5/3, 6/1, 6/4, 7/2, 7/3, 8/1, 8/4}
        \draw[->]  (y\n) -- (z\m);

        \foreach \n\m in {1/1, 2/1, 3/2, 4/2, 5/3, 6/3, 7/4, 8/4,
            9/5, 10/6, 11/7, 12/8}
        \draw[->] (v_\n) -- (y\m);

    \end{tikzpicture}
\end{center}

Below, we~show that CNF encodings and depth-3 circuits can be~easily transformed one
into the other. It~will prove convenient to~define the size of a~circuit as~its number of~gates \emph{excluding} the output gate. This way, the size of
a~CNF formula equals its number of~clauses (a~CNF is a~depth-2 formula).
By~a~$\Sigma_3(t,r)$-circuit we~denote a~$\Sigma_3$-circuit
having at~most $t$~ANDs on the second layer and at~most $r$~ORs
on~the third layer (hence, its size is at~most $t+r$).

\begin{lemma}
    Let $F(x_1, \dotsc, x_n, y_1, \dotsc, y_s)$ be a~CNF encoding of~size~$m$
    of a~function
    $f \colon \{0,1\}^n \to \{0,1\}$.
    Then, $f$~can be~computed
    by a~$\Sigma_3(2^s, m \cdot 2^s)$-formula and by a~$\Sigma_3(2^s,m)$-circuit.
\end{lemma}
\begin{proof}
    Let $F=\{C_1, \dotsc, C_m\}$. To~expand~$F$ as~$\bigvee_{j \in [2^s]}F_j$,
    we~go~through all $2^s$ assignments to non-deterministic variables
    $y_1, \dotsc, y_s$. Under any such assignment, each clause $C_i$
    is~either satisfied or~becomes a~clause $C_i' \subseteq C_i$
    resulting from $C_i$ by~removing all its non-deterministic variables.
    Thus, for each $j \in [2^s]$, $F_j \subseteq \{C_1', \dotsc, C_m'\}$.
    The corresponding $\Sigma_3$-formula contains at~most $2^s+m2^s$ gates:
    there are $2^s$ gates for $F_j$'s, each $F_j$ contains no more than $m$~clauses.
    The corresponding $\Sigma_3$-circuit contains no~more than $2^s+m$ gates:
    there are $2^s$ gates for $F_j$'s and $m$~gates for $C_1', \dotsc, C_m'$
    (each $F_j$ selects which of these $m$~clauses to~contain).
\end{proof}

Interestingly, the upper bounds on~depth-3 circuits resulting from this simple transformation cannot be~substantially improved. Indeed, by~plugging~in
a~CNF encoding of~$\PAR_n$ with $s=\sqrt n$ and $m=O(\sqrt n2^{\sqrt n})$
(see~\eqref{eq:blocks}),
one gets a~$\Sigma_3$-formula and a~$\Sigma_3$-circuit of size $2^{2\sqrt n}$
and $2^{\sqrt n}$, respectively, up~to polynomial factors. As~discussed above,
these bounds are known to~be optimal.

Below, we~show a~converse transformation.

\begin{lemma}\label{lemma:circuit2encoding}
    Let $C$~be a~$\Sigma_3(t,r)$-formula (circuit) computing a~Boolean function $f \colon \{0,1\}^n \to \{0,1\}$. Then,
    $f$~can be~encoded as a~CNF with $\lceil \log t \rceil$ non-deterministic variables of size~$r$ ($2rt$, respectively).
\end{lemma}

\begin{proof}
    Let $C=F_1 \lor \dotsb \lor F_{t}$ be a~$\Sigma_3$-formula (hence, $r=\size(F_1)+\dotsb+\size(F_{t})$). Introduce $s=\lceil\log t\rceil$ non-deterministic variables $y_1, \dotsc, y_s$. Then, for every assignment to $y_1, \dotsc, y_s$, take the corresponding CNF $F_i$ ($1 \le i \le 2^s$ is the unique integer corresponding to this assignment) and add
    $y_i$'s with the corresponding signs to~every clause of~$F_i$.
    Call the resulting CNF $F_i'$. Then, $F=F_1'\land \dotsb \land F_{2^s}'$ encodes~$f$ and $F$~has at most $r$~clauses.

    If $C$~is a~$\Sigma_3$-circuit, we~need to~create a~separate copy
    of~every gate corresponding to a~clause in each of~$2^s$ CNFs. Hence, the size
    of~the resulting CNF encoding is at~most $r2^s \le 2rt$.
\end{proof}

Finally, we~show that proving strong lower bounds on~the size of CNF encodings is~not easier than proving strong lower bounds on~the size of~depth-3 circuits.
Let $C$~be a~$\Sigma_3(t,r)$-formula computing~$\PAR_n$.
Lemma~\ref{lemma:circuit2encoding} guarantees that $\PAR_n$ can be~encoded as a~CNF of size~$r$ with $\lceil \log t \rceil$ non-deterministic variables. Then, by~the
inequality~\eqref{eq:sm},
\[\size(C) = t + r \ge t + \Omega \left( \frac{1}{n}\cdot 2^{\frac{n}{\log t + 2}} \right) \ge \frac 1n \left( t + \Omega\left( 2^{\frac{n}{\log t + 2}}\right)\right) \ge \Omega\left(\frac{2^{\sqrt{n}}}{n}\right) \, .\]
Similarly, if $C$~is a~$\Sigma_3(t,r)$-circuit, Lemma~\ref{lemma:circuit2encoding} guarantees that $\PAR_n$ can be~encoded as a~CNF of size~$2rt$ with $\lceil \log t \rceil$ non-deterministic variables. Then,
\[\size(C) = t + r \ge t + \Omega \left( \frac{1}{2tn}\cdot 2^{\frac{n}{\log t + 2}} \right) \ge \Omega\left(\frac{2^{\sqrt{n/2}}}{n}\right) \, .\]

\section{Lower bounds for CNF encodings of parity}
In~this section, we~prove Theorem~\ref{thm:main}. The essential
property of the parity function used in~the proof
is~its high sensitivity (every satisfying assignment is isolated): for any $i \in [n]$
and any $x,x' \in \{0,1\}^n$ that differ in the $i$-th position only,
$\PAR(x) \neq \PAR(x')$. This means that if a~CNF~$F$
computes $\PAR$ and $F(x)=1$, then $F$~must contain
a~clause that is satisfied by~$x_i$ only.
Following~\cite{DBLP:journals/cjtcs/PaturiPZ99},
we~call such a~clause \emph{critical} with respect to~$(x,i)$.
This notion extends to~CNF encodings in a~natural way.
Namely, let $F(x, y)$ be a~CNF encoding of~$\PAR$.
Then, for any $(x,y)$ such that $F(x,y)=1$ and any $i \in [n]$,
$F$~contains a~clause that becomes falsified if~one flips the bit~$x_i$.
We~call~it critical w.r.t. $(x,y,i)$.

\subsection{Limited non-determinism}
To~prove a~lower bound $m \ge \Omega((s+1)2^{n/(s+1)}/n)$, we adapt a~proof of the $\Omega(n^{1/4}2^{\sqrt n})$ lower bound for depth-3 circuits computing $\PAR_n$ by Paturi, Pudl\'{a}k, and Zane~\cite{DBLP:journals/cjtcs/PaturiPZ99}. Let $F(x_1, \dotsc, x_n)$ be a~CNF.
For every isolated satisfying assignment~$x \in \{0,1\}^n$ of~$F$ and every~$i \in [n]$, fix a~shortest critical clause w.r.t. $(x,i)$ and denote~it by~$C_{F,x,i}$.
Then, for an~isolated satisfying assignment~$x$, define its weight w.r.t.~$F$~as
\[w_F(x) = \sum\limits_{i=1}^n \frac{1}{|C_{F,x,i}|} \, .\]

\begin{lemma}[Lemma~5 in~\cite{DBLP:journals/cjtcs/PaturiPZ99}]\label{lemma:isolatedweight}
    For any~$\mu$, $F$~has at~most $2^{n - \mu}$ isolated satisfying assignments
    of~weight at~least~$\mu$.
\end{lemma}

\begin{proof}[Proof of~\eqref{eq:sm}, $m \ge \Omega\left(\frac{s+1}{n} \cdot 2^{n/(s+1)}\right)$]
    Let $F(x_1, \dotsc, x_n, y_1, \dotsc, y_s)$ be a~CNF encoding of size~$m$ of $\PAR_n$. Consider its expansion:
    \[\PAR_n(x)=\bigvee_{j \in [2^s]} F_j(x) \, .\]

    We extend the definitions of $C_{F,x,i}$ and $w(x)$ to CNFs with non-deterministic variables as~follows. Let $x \in \PAR^{-1}_n(1)$
    and let $j \in [2^s]$ be the smallest index such that $F_j(x)=1$.
    For $i \in [n]$, let $C'_{F,x,i}=C_{F_j,x,i}$ (that is, we~simply take the
    first $F_j$ that is~satisfied by~$x$ and take its critical clause w.r.t. $(x,i)$).
    Then, the weight $w'_F(x)$ of~$x$ w.r.t. to~$F$ is~defined simply as~$w_{F_j}(x)$.
    Clearly,
    \[w'_F(x) = \sum_{i \in [n]} \frac{1}{|C'_{(F,x,i)}|} \,. \]
    For $l \in [n]$, let also $N_{l,F}(x)=|\{i \in [n] \colon |C'_{F,x,i}|=l\}|$
    be~the number of~critical clauses (w.r.t.~$x$) of~length~$l$. Clearly,
    \begin{equation}\label{eq:weight}
        w'_F(x)=\sum_{l \in [n]}\frac{N_{l,F}(x)}{l} \, .
    \end{equation}

    For a~parameter $0<\varepsilon<1$ to be~chosen later, split $\PAR_n^{-1}(1)$
    into light and heavy parts:
    \begin{align*}
        H &= \{x \in \PAR_n^{-1}(1)  \colon w'_F(x) \ge s + 1 + \varepsilon\} \, ,\\
        L &= \{x \in \PAR_n^{-1}(1)  \colon w'_F(x) < s + 1 + \varepsilon\} \, .
    \end{align*}

    We~claim that
    \[|H| \le 2^s \cdot 2^{n-s-1-\varepsilon} \, .\]
    Indeed, for every $x \in H$, $w'_F(x)=w_{F_j}(x)$ for some $j \in [2^s]$,
    and by~Lemma~\ref{lemma:isolatedweight}, $F_j$~cannot accept more than
    $2^{n-s-1-\varepsilon}$ isolated solutions of weight at~least
    $s+1+\varepsilon$. Since $|H|+|L|=|\PAR^{-1}_n(1)|=2^{n-1}$, we~conclude that
    \begin{equation}\label{eq:lightsize}
        |L|=2^{n-1} -|H| \ge (1 - 2^{-\varepsilon})2^{n-1} \, .
    \end{equation}

    Let $F=\{C_1, \dotsc, C_m\}$. For every $k \in [m]$, let $C'_k \subseteq C_k$
    be the clause~$C_k$ with all non-deterministic variables removed. Hence, for every
    $j \in [2^s]$, $F_j \subseteq \{C_1', \dotsc, C_m'\}$. For $l \in [n]$,
    let $m_l=|\{k \in [m] \colon |C'_k|=l\}|$ be~the number of~such clauses of length~$l$.
    Consider a~clause~$C'_k$
    and let $l=|C'_k|$. Then, there are at~most $l2^{n-l}$ pairs $(x,i)$,
    where $x \in \PAR^{-1}(1)$ and $i \in [n]$, such that $C'_{F,x,i}=C_k'$:
    there are at~most~$l$ choices for~$i$, fixing~$i$ fixes the values
    of~all $l$~literals in~$C_k'$ (all of them are equal to~zero except for
    the $i$-th one), and there are no~more than $2^{n-l}$ choices for the other bits of~$x$.
    Recall that $N_{l,F}(x)$ is the number of~critical clauses w.r.t.~$x$
    of~length~$l$. Thus, we~arrive at~the following inequality:
    \[m_l \cdot l \cdot 2^{n - l} \ge \sum_{x \in \PAR^{-1}(1)} N_{F,l}(x) \ge \sum_{x \in L} N_{F,l}(x) \, .\]
    Then,
    \begin{equation}\label{eq:cc}
        m =\sum_{l \in [n]}m_l \ge \sum_{l \in [n]}\frac{\sum_{x \in L} N_{F,l}(x)}{l2^{n-l}}=
        \sum_{x \in L} \sum_{l\in [n]} \frac{N_{F,l}(x)}{l2^{n-l}} =
        \sum_{x \in L} n2^{-n} \sum_{l \in [n]} \frac{N_{F,l}(x)}{n} \cdot \frac{2^l}{l} \, .
    \end{equation}

    To~estimate the last sum, let
    \[T(x)=\sum_{l \in [n]} \frac{N_{F,l}(x)}{n} \cdot \frac{2^l}{l}=\sum_{l \in [n]} \frac{N_{F,l}(x)}{n}\cdot g(l)\,,\]
    where $g(l) = \frac{2^l}{l}$. Since $g(l)$ is~convex (for $l>0$) and $\sum_{l \in [n]}\frac{N_{F,l}(x)}{n}=1$, Jensen's inequality gives
    \begin{equation}\label{eq:aa}
        T(x) \ge g\left(\sum_{l \in [n]} \frac{N_{F,l}(x)}{n}\cdot l\right) \, .
    \end{equation}
    Further, Sedrakyan's inequality\footnote{Sedrakyan's inequality is a~special case of Cauchy--Schwarz inequality: for all $a_1, \dotsc, a_n \in \mathbb R$ and $b_1, \dotsc, b_n \in \mathbb{R}_{>0}$, $\sum_{i=1}^n a_i^2/b_i \ge \left(\sum_{i=1}^n a_i\right)^2/\sum_{i=1}^n b_i$.} (combined with \eqref{eq:weight} and $\sum_{l \in [n]}N_{F,l}(x)=n$) gives
    \begin{equation}\label{eq:bb}
        \sum_{l \in [n]} l N_{F,l}(x) = \sum_{l \in [n]} \frac{N_{F,l}^2(x)}{N_{F,l}(x) / l} \ge \frac{(\sum_{l \in [n]} N_{F,l}(x))^2}{\sum_{l \in [n]} N_{F,l}(x) / l} = \frac{n^2}{w'_F(x)} \, .
    \end{equation}
    Since $g(l)$ is monotonically increasing for $l \ge 1/\ln 2$ and $w'_F(x)<s+1+\varepsilon$ for every $x \in L$,
    combining \eqref{eq:aa}~and~\eqref{eq:bb}, we~get
    \begin{equation}\label{eq:tx}
        T(x) \ge g \left( \frac{n}{w'_F(x)}\right) \ge g \left( \frac{n}{s+1+\varepsilon}\right) \, ,
    \end{equation}
    for $s \le n\ln 2 - 1 -\varepsilon$.
    (If $s > n\ln 2 - 1 -\varepsilon$, then the lower bound $m \ge \Omega(2^{n/(s+1)}/n)$ is trivial.)

    Thus,
    \begin{align*}
        m &\ge \sum_{x \in L}n2^{-n}T(x) \ge \tag{\ref{eq:cc}~and~\ref{eq:tx}}\\
          & \ge \sum_{x \in L}n2^{-n}g \left( \frac{n}{s+1+\varepsilon}\right)=\tag{definition of~$g$}\\
          &=|L|2^{-n}2^{\frac{n}{s+1+\varepsilon}}(s+1+\varepsilon) \ge\tag{\ref{eq:lightsize}}\\
          &\ge \left(\frac{1}{2} - \frac{1}{2^{\varepsilon + 1}}\right)(s + 1 + \varepsilon) 2^\frac{n}{s + 1 + \varepsilon}=\tag{rewriting}\\
          &=\left(\frac{1}{2} - \frac{1}{2^{\varepsilon + 1}}\right)(s + 1 + \varepsilon) 2^\frac{n}{s + 1} 2^{\frac{-n\varepsilon}{(s+1)(s+1+\varepsilon)}} \,.
    \end{align*}

    Set $\varepsilon=1/n$. Then,
	\[\left(\frac{1}{2} - \frac{1}{2^{\frac{1}{n} + 1}}\right) = \Theta \left( \frac{1}{n} \right)\,.\]
    Also,
    \[\frac{1}{2} \le 2^{\frac{-1}{(s+1)(s+1+1/n)}} \le 1\,,\]
    as $2^{-1/x}$ is~increasing for $x>0$.
	This finally gives a~lower bound
	\[m \ge \Omega \left(\frac{s+1}{n} \cdot 2^{\frac{n}{s+1}} \right)\, .\]
\end{proof}

\subsection{Width of clauses}
To~prove the lower bound $k \ge n/(s+1)$, we~use the following corollary of~the Satisfiability Coding Lemma.

\begin{lemma}[Lemma~2 in~\cite{DBLP:journals/cjtcs/PaturiPZ99}]\label{lemma:isolated}
    Any $k$-CNF $F(x_1, \dotsc, x_n)$ has at~most $2^{n - n/k}$ isolated satisfying assignments.
\end{lemma}

\begin{proof}[Proof of~\eqref{eq:sw}, $k \ge n/(s+1)$]
	Consider a~$k$-CNF $F(x_1, \dotsc, x_n, y_1, \dotsc, y_s)$ that encodes $\PAR_n$. Expand~$F$ to an~OR of $2^s$ $k$-CNFs:
	\[\PAR_n(x)=\bigvee_{j \in [2^s]} F_j(x) \, .\]
	By Lemma~\ref{lemma:isolated}, each $F_j$ accepts at~most $2^{n - n/k}$ isolated solutions. Hence,
	\[2^s \ge \frac{2^{n - 1}}{2^{n - n/k}} = 2^{n/k - 1}\]
    and thus, $k \ge n/(s+1)$.
\end{proof}

\subsection{Unlimited non-determinism}
In this section, we prove the lower bound $m \ge 3n-9$.

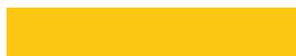
\begin{proof}[Proof of~\eqref{eq:m}, $m \ge 3n-9$]
    We~use induction on~$n$.
    The base case $n \le 3$ is clear.
    To~prove the induction step, assume that $n>3$ and consider a~CNF encoding~$F(x_1, \dotsc, x_n, y_1, \dotsc, y_s)$ of $\PAR_n$ with the minimum number of~clauses. Below, we~show
    that one can find $k$~deterministic variables (where $k=1$ or~$k=2$) such that assigning appropriately chosen constants to~them reduces the number of~clauses
    by~at~least $3k$, respectively. The resulting function computes
    $\PAR_{n-k}$ or its negation. It~is not difficult to~see
    that the minimum number of~clauses in~encodings of~$\PAR$
    and its negation are equal (by~flipping the signs of~all
    occurrences of any deterministic variable in a~CNF encoding of~$\PAR$,
    one gets a~CNF encoding of the negation of~$\PAR$, and vice versa). Hence,
    one can proceed by~induction and conclude that $F$~contains at~least
    $3(n-k)-9+3k=3n-9$ clauses.

    To~find the required $k$~deterministic variables,
    we~go~through a~number of~cases. In the analysis below,
    by a~$d$-literal we~mean a~literal that appears exactly $d$~times
    in~$F$, a~$d^+$-literal appears at~least $d$~times. A~$(d_1,d_2)$-literal
    occurs $d_1$~times positively and $d_2$~times negatively. Other types
    of~literals are defined similarly. We~treat a~clause as~a~set of~literals
    (that do~not contain a~literal together with its negation) and a~CNF
    formula as a~set of~clauses.

    Note that for all $i \in [s]$, $y_i$ must be a~$(2^+,2^+)$-literal. Indeed,
    if $y_i$ (or $\overline{y_i}$) is a~$0$-literal, one can assign $y_i \gets 0$
    ($y_1 \gets 1$, respectively). It~is not difficult to~see that the resulting formula still encodes $\PAR$. If $y_i$ is a~$(1,t)$-literal, one can
    eliminate~it using resolution: for all pairs of clauses $C_0, C_1 \in F$ such that
    $\overline{y_i} \in C_0$ and $y_i \in C_1$, add a~clause $C_0 \cup C_1 \setminus \{y_i, \overline{y_i}\}$ (if~this clause contains a~pair of~complementary
    literals, ignore~it); then, remove all clauses containing $y_i$ or~$\overline{y_i}$. The resulting formula still encodes $\PAR_n$,
    but has a~smaller number of~clauses than~$F$ (we~remove $1+t$ clauses and add at~most $t$~clauses).

    In~the case analysis below,
    by~$l_i$ we~denote a~literal that corresponds
    to~a~deterministic variable~$x_i$ or~its negation~$\overline{x_i}$.
    \begin{enumerate}
        \item \emph{$F$~contains a~$3^+$-literal~$l_i$.} Assigning $l_i \gets 1$ eliminates at~least three clauses from~$F$.

        \item \emph{$F$~contains a~$1$-literal~$l_i$.} Let $l_i \in C \in F$ be a~clause containing~$l_i$. $C$~cannot contain other deterministic variables: if $l_i, l_j \in C$ (for $i \neq j \in [n]$), consider $x \in \{0,1\}^n$ such that
        $\PAR_n(x)=1$ and $l_i=l_j=1$ (such~$x$ exists since $n > 3$), and its extension $y \in \{0,1\}^s$ such that $F(x,y)=1$; then, $F$~does not contain
        a~critical clause w.r.t. $(x,y,i)$. Clearly, $C$~cannot be a~unit clause, hence it~must contain a~non-deterministic variable~$y_j$. Consider $x \in \{0,1\}^n$, such that $\PAR_n(x)=1$ and $l_i=1$, and its extension $y \in \{0,1\}^s$ such that $F(x,y)=1$. If $y_j=1$, then $F$~does not contain
        a~critical clause w.r.t. $(x,y,i)$. Thus, for every $(x,y) \in \{0,1\}^{n+s}$ such that $F(x,y)=1$ and $l_i=1$, it holds that $y_j=0$. This observation allows~us to~proceed as~follows: first assign $l_i \gets 1$, then assign $y_j \gets 0$. The former assignment satisfies the clause~$C$, the latter one satisfies all the clauses containing $\overline{y_j}$. Thus, at~least three clauses are removed.

        \item\label{case:three} \emph{For all $i \in [n]$, $x_i$ is a~$(2,2)$-literal.} If there~is
        no~clause in~$F$ containing at~least two deterministic variables, then
        $F$~contains at~least $4n$ clauses and there~is nothing to~prove.
        Let $l_i,l_j \in C_1 \in F$, where $i \neq j$, be a~clause
        containing two deterministic variables and let $l_i \in C_2 \in F$ and $l_j \in C_3 \in F$ be the two clauses containing other occurrences of~$l_i$ and~$l_j$ ($C_1 \neq C_2$ and $C_1 \neq C_3$, but it can be the case that $C_2=C_3$).

        Assume that $C_2$ contains another deterministic variable: $l_k \in C_2$, where $k \neq i,j$.
        Consider $x \in \{0,1\}^n$, such that $\PAR_n(x)=1$ and $l_i=l_j=l_k=1$
        (such $x$~exists since $n>3$), and its extension $y \in \{0,1\}^s$ such that $F(x,y)=1$. Then, $F$~does not contain a~critical clause w.r.t. $(x,y,i)$:
        $C_1$ is~satisfied by~$l_j$, $C_2$ is~satisfied by~$l_k$. For the same reason, $C_2$ cannot contain the literal~$l_j$. Similarly, $C_3$ cannot contain other deterministic variables and the literal~$l_i$.
        (At the same time, it is not excluded that $\overline{l_j} \in C_2$ or
        $\overline{l_i} \in C_3$.)
        Hence, $C_2 \neq C_3$. Note that each
        of~$C_2$ and~$C_3$ must contain at~least one non-deterministic variable:
        otherwise, it~would be~possible to~falsify~$F$ by~assigning
        $l_i$~and~$l_j$.

        \begin{enumerate}
            \item \emph{At~least one of~$C_2$ and~$C_3$ contains a~single non-deterministic variable.} Assume that it is~$C_2$:
            \[\{l_i,y_1\} \subseteq C_2 \subseteq \{l_i, \overline{l_j}, y_1\}\,. \]
            Assign $l_j \gets 1$. This eliminates two clauses: $C_1$ and $C_3$ are satisfied. Also, under this substitution, $C_2=\{l_i,y_1\}$ and $l_i$
            is a~$1$-literal. We claim that in any satisfying assignment of the resulting formula~$F'$, $l_i=\overline{y_1}$. Indeed, if $(x,y)$
            satisfies $F'$ and $l_i=y_1$, then $l_i=y_1=1$ (otherwise $C_2$
            is~falsified). But then there is no critical clause in~$F'$ w.r.t. $(x,y,i)$. Since in every satisfying assignment $l_i=\overline{y_1}$,
            we~can replace every occurrence of~$y_1$ ($\overline{y_1}$)
            by~$\overline{l_i}$ ($y_1$, respectively). This, in~particular, satisfies
            the clause~$C_2$.

            \item \emph{Both $C_2$ and $C_3$ contain at~least two non-deterministic variables:}
            \[
            \{l_i,\ l_j\} \subseteq C_1, \quad
            \{l_i,\ y_1,\ y_2\} \subseteq C_2, \quad
            \{l_j,\ y_3,\ y_4\} \subseteq C_3 \, .
            \]
            Here, $y_1$~and~$y_2$ are different variables, $y_3$~and~$y_4$ are also different, though it~is not excluded that some of $y_1$~and~$y_2$
            coincide with some of $y_3$~and~$y_4$. Let $Y \subseteq \{y_1, \dotsc, y_s\}$ be~non-deterministic variables appearing in~$C_2$ or~$C_3$.

            Recall that for every $(x,y) \in \{0,1\}^{n+s}$ such that $F(x,y)=1$ and $l_i=l_j = 1$, it holds that $y=0$ for all $y \in Y$. This means that
            if a~variable $y \in Y$ appears in both~$C_2$ and~$C_3$, then it has the same sign in~both clauses. Consider two subcases.

            \begin{enumerate}
                \item $Y=\{y_1,y_2\}$: \[\{l_i,\ l_j\} \subseteq C_1,\quad  \{l_i,\ y_1,\ y_2\} \subseteq C_2,\quad  \{l_j,\ y_1,\ y_2\} \subseteq C_3 \, .\]

                Assume that $\overline{y_1} \not \in C_1$. Assign $l_i \gets 1$, $l_j \gets 1$. Then, assigning $y_1 \gets 0$ eliminates at~least two clauses. Let~us show that there remains
                a~clause that contains $\overline{y_2}$. Consider $x \in \PAR_n^{-1} (1)$, such that $l_i = l_j = 1$, and its extension $y \in \{0, 1\}^s$, such $F(x, y) = 1$. We~know that $y_1$ and $y_2$ must
                be~equal to~$0$. However, flipping the value of~$y_2$ results
                in~a~satisfying assignment.
                Thus, it remains to analyze the following case: \[\{l_i,\ l_j, \overline{y_1}, \overline{y_2} \} \subseteq C_1, \quad               \{l_i,\ y_1,\ y_2\} \subseteq C_2, \quad
                \{l_j,\ y_1,\ y_2\} \subseteq C_3 \, .\]

                Assume that $\overline{l_j} \not \in C_2$ and $\overline{l_i} \not \in C_1$. Assign $l_i \gets 1$, then assign $y_1 \gets 0$ and $y_2 \gets 0$. Under this assignment, $C_3 = \{l_j\}$ (recall that $C_3$ cannot contain other deterministic variables, see Case~\ref{case:three}). This would mean that $l_j=1$ in~every satisfying assignment of the resulting CNF formula which cannot be~the case for a~CNF encoding of~parity.
                Thus, we~may assume that either $\overline{l_j} \in C_2$ or $\overline{l_i} \in C_1$. Without loss of~generality, assume that
                $\overline{l_j} \in C_2$.

                Let~us show that for every $(x, y) \in \{0, 1\}^{n + s}$, such that $F(x, y) = 1$ and $l_i = 1$, it~holds that $l_j \neq y_1$ and $l_j \neq y_2$. Indeed, if there is $(x, y) \in \{0, 1\}^{n + s}$ such that $F(x,y)=1$ and $l_i = l_j = 1$, then $y_1$ and $y_2$ must be equal to~$0$. If there is $(x, y) \in \{0, 1\}^{n + s}$, such that $F(x, y) = 1, l_i = 1, l_j = 0$, then $y_1$ and $y_2$ must be equal to~$0$, otherwise $F$~does not contain a~critical clause w.r.t. $(x, y, i)$. Thus, assigning $l_i \gets 1$ eliminates two clauses ($C_1$~and~$C_2$). We~then replace~$y_1$ and~$y_2$ with $\overline{l_j}$ and delete the clause~$C_3$.

                \item $|Y| \ge 3, \{y_1,y_2,y_3\} \subseteq Y$: \[\{l_i,\ l_j\} \subseteq C_1,\quad  \{l_i,\ y_1,\ y_2\} \subseteq C_2,\quad  \{l_j,\ y_1,\ y_3\} \subseteq C_3 \, .\]

                Assigning $l_i \gets 1, l_j \gets 1$ eliminates $C_1, C_2, C_3$. Assigning $y_1 \gets 0$ eliminates at least one more clause
                ($y_1$ appears positively at least two times, but it~may appear in~$C_1$).
                There must be a~clause with $\overline{y_2}$ (otherwise we could assign $y_2 \gets 1$). Assigning $y_2 \gets 0$ eliminates at~least one more clause. Similarly, assigning $y_3 \gets 1$ eliminates another clause. In total, we eliminate at~least six clauses.
            \end{enumerate}
        \end{enumerate}
       \end{enumerate}
\end{proof}

\section*{Acknowledgments}
Research is partially supported by Huawei (grant TC20211214628).

\bibliography{references}

\begin{thebibliography}{10}

\bibitem{DBLP:journals/siamcomp/AllenderHMPS08}
Eric Allender, Lisa Hellerstein, Paul McCabe, Toniann Pitassi, and Michael~E.
  Saks.
\newblock Minimizing disjunctive normal form formulas and $\operatorname{AC}^0$
  circuits given a truth table.
\newblock {\em {SIAM} J. Comput.}, 38(1):63--84, 2008.
\newblock \href {https://doi.org/10.1137/060664537}
  {\path{doi:10.1137/060664537}}.

\bibitem{DBLP:journals/sigact/GoldsmithLM96}
Judy Goldsmith, Matthew~A. Levy, and Martin Mundhenk.
\newblock Limited nondeterminism.
\newblock {\em {SIGACT} News}, 27(2):20--29, 1996.
\newblock \href {https://doi.org/10.1145/235767.235769}
  {\path{doi:10.1145/235767.235769}}.

\bibitem{DBLP:journals/eccc/Hirahara17}
Shuichi Hirahara.
\newblock A duality between depth-three formulas and approximation by
  depth-two.
\newblock {\em Electron. Colloquium Comput. Complex.}, page~92, 2017.
\newblock URL: \url{https://eccc.weizmann.ac.il/report/2017/092}.

\bibitem{DBLP:books/daglib/0028687}
Stasys Jukna.
\newblock {\em Boolean Function Complexity - Advances and Frontiers}, volume~27
  of {\em Algorithms and combinatorics}.
\newblock Springer, 2012.
\newblock \href {https://doi.org/10.1007/978-3-642-24508-4}
  {\path{doi:10.1007/978-3-642-24508-4}}.

\bibitem{DBLP:journals/tcs/KuceraSV19}
Petr Kucera, Petr Savick{\'{y}}, and Vojtech Vorel.
\newblock A lower bound on {CNF} encodings of the at-most-one constraint.
\newblock {\em Theor. Comput. Sci.}, 762:51--73, 2019.
\newblock \href {https://doi.org/10.1016/j.tcs.2018.09.003}
  {\path{doi:10.1016/j.tcs.2018.09.003}}.

\bibitem{MasekNpComp}
William~J. Masek.
\newblock Some {NP}-complete set covering problems.
\newblock Unpublished Manuscript, 1979.

\bibitem{DBLP:conf/cocoon/Morizumi15}
Hiroki Morizumi.
\newblock Lower bounds for the size of nondeterministic circuits.
\newblock In Dachuan Xu, Donglei Du, and Ding{-}Zhu Du, editors, {\em Computing
  and Combinatorics - 21st International Conference, {COCOON} 2015, Beijing,
  China, August 4-6, 2015, Proceedings}, volume 9198 of {\em Lecture Notes in
  Computer Science}, pages 289--296. Springer, 2015.
\newblock \href {https://doi.org/10.1007/978-3-319-21398-9\_23}
  {\path{doi:10.1007/978-3-319-21398-9\_23}}.

\bibitem{DBLP:journals/cjtcs/PaturiPZ99}
Ramamohan Paturi, Pavel Pudl{\'{a}}k, and Francis Zane.
\newblock Satisfiability coding lemma.
\newblock {\em Chic. J. Theor. Comput. Sci.}, 1999, 1999.
\newblock URL:
  \url{http://cjtcs.cs.uchicago.edu/articles/1999/11/contents.html}.

\bibitem{DBLP:journals/dam/Prestwich03}
Steven~David Prestwich.
\newblock {SAT} problems with chains of dependent variables.
\newblock {\em Discret. Appl. Math.}, 130(2):329--350, 2003.
\newblock \href {https://doi.org/10.1016/S0166-218X(02)00410-9}
  {\path{doi:10.1016/S0166-218X(02)00410-9}}.

\bibitem{DBLP:series/faia/Prestwich09}
Steven~David Prestwich.
\newblock {CNF} encodings.
\newblock In Armin Biere, Marijn Heule, Hans van Maaren, and Toby Walsh,
  editors, {\em Handbook of Satisfiability}, volume 185 of {\em Frontiers in
  Artificial Intelligence and Applications}, pages 75--97. {IOS} Press, 2009.
\newblock \href {https://doi.org/10.3233/978-1-58603-929-5-75}
  {\path{doi:10.3233/978-1-58603-929-5-75}}.

\bibitem{DBLP:conf/cp/Sinz05}
Carsten Sinz.
\newblock Towards an optimal {CNF} encoding of boolean cardinality constraints.
\newblock In Peter van Beek, editor, {\em Principles and Practice of Constraint
  Programming - {CP} 2005, 11th International Conference, {CP} 2005, Sitges,
  Spain, October 1-5, 2005, Proceedings}, volume 3709 of {\em Lecture Notes in
  Computer Science}, pages 827--831. Springer, 2005.
\newblock \href {https://doi.org/10.1007/11564751\_73}
  {\path{doi:10.1007/11564751\_73}}.

\bibitem{zbMATH03325539}
G.~S. {Tsejtin}.
\newblock {On the complexity of derivation in propositional calculus}.
\newblock {Semin. Math., V. A. Steklov Math. Inst., Leningrad 8, 115-125
  (1970); translation from Zap. Nauchn. Semin. Leningr. Otd. Mat. Inst.
  Steklova 8, 234-259 (1968).}, 1968.

\bibitem{DBLP:conf/mfcs/Valiant77}
Leslie~G. Valiant.
\newblock Graph-theoretic arguments in low-level complexity.
\newblock In Jozef Gruska, editor, {\em Mathematical Foundations of Computer
  Science 1977, 6th Symposium, Tatranska Lomnica, Czechoslovakia, September
  5-9, 1977, Proceedings}, volume~53 of {\em Lecture Notes in Computer
  Science}, pages 162--176. Springer, 1977.
\newblock \href {https://doi.org/10.1007/3-540-08353-7\_135}
  {\path{doi:10.1007/3-540-08353-7\_135}}.

\end{thebibliography}

\end{document}